\documentclass[a4paper, 10pt, twocolumn]{article}
\usepackage[utf8]{inputenc}
\usepackage[english]{babel}
\usepackage{mathtools}
\usepackage{amsmath,amssymb,amsfonts}
\usepackage{amsthm}
\usepackage{algpseudocode}
\usepackage{algorithm}
\usepackage{dirtytalk}
\usepackage{csquotes}
\usepackage{url}
\usepackage{microtype}
\usepackage{mathpazo}
\usepackage{authblk}
\usepackage{biblatex}
\newtheorem{definition}{Definition}
\newtheorem{theorem}{Theorem}

\newtheorem{corollary}{Corollary}

\title{Generalised Likelihood Ratio Testing Adversaries through the Differential Privacy Lens}

\author[1, 2, 4]{Georgios Kaissis \thanks{Corresponding author e-mail: g.kaissis@tum.de}}
\author[1, 4]{Alexander Ziller}
\author[3, 4]{Stefan Kolek Martinez de Azagra}
\author[1, 2]{Daniel Rueckert}

\affil[1]{Artificial Intelligence in Medicine and Healthcare, Technical University of Munich}
\affil[2]{Department of Computing, Imperial College London}
\affil[3]{Mathematical Foundations of Artificial Intelligence, LMU Munich}
\affil[4]{These authors contributed equally}

\date{}

\begin{document}
\maketitle

\begin{abstract}
Differential Privacy (DP) provides tight upper bounds on the capabilities of optimal adversaries, but such adversaries are rarely encountered in practice. Under the hypothesis testing/membership inference interpretation of DP, we examine the Gaussian mechanism and relax the usual assumption of a Neyman-Pearson-Optimal (NPO) adversary to a Generalized Likelihood Test (GLRT) adversary. This mild relaxation leads to improved privacy guarantees (see Figure \ref{fig:teaser} below), which we express in the spirit of Gaussian DP and $(\varepsilon, \delta)$-DP, including composition and sub-sampling results. We evaluate our results numerically and find them to match the theoretical upper bounds.
\end{abstract}

\begin{figure}[h!]
    \centering
    \includegraphics[width=0.45\textwidth]{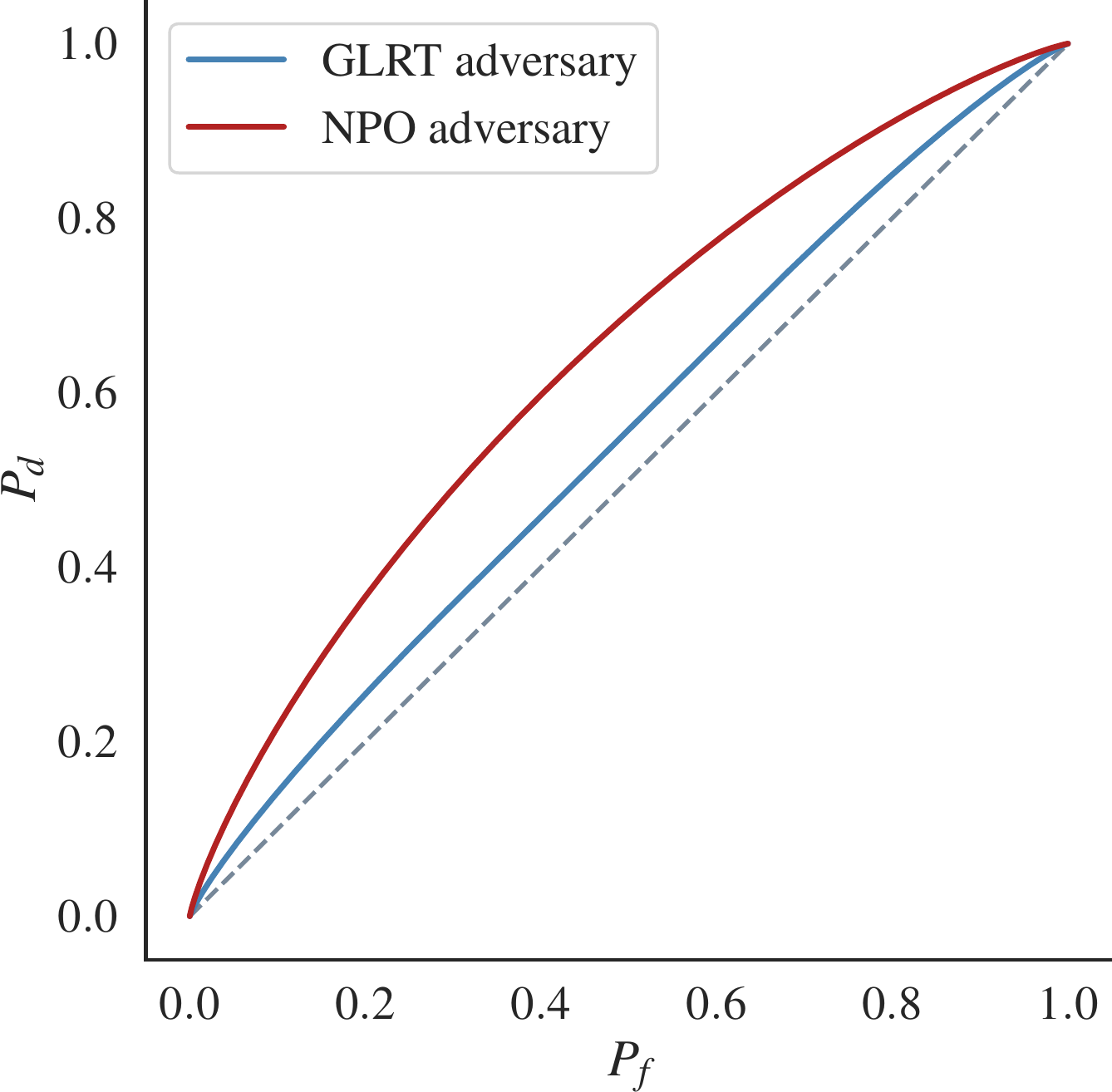}
    \caption{Our results at a glance: A minimal relaxation of the threat model from an NPO adversary (red curve) to a GLRT adversary leads to a substantially more optimistic outlook on privacy loss (blue curve) from $(\varepsilon, \delta)=(0.95, 10^{-4})$-DP to $(0.37, 10^{-4})$-DP at $\Delta/\sigma=0.5$. $P_{d/f}$: Probability of detection/false-positive.}
    \label{fig:teaser}
\end{figure}

\section{Introduction}

Differential Privacy (DP) and its applications to machine learning (ML) have established themselves as the tool of choice for statistical analyses on sensitive data. They allow analysts working with such data to obtain useful insights while offering objective guarantees of privacy to the individuals whose data is contained within the dataset. DP guarantees are typically realised through the addition of calibrated noise to statistical queries. The randomisation of queries however introduces an unavoidable \say{tug-of-war} between privacy and accuracy, the so-called \textit{privacy-utility trade-off}. This trade-off is undesirable and may be among the principal deterrents from the widespread willingness to commit to the usage of DP in statistical analyses.

The main reason why DP is considered harmful for utility is perhaps an incomplete understanding of its very formulation: In its canonical definition, DP is a worst-case guarantee against a very powerful (i.e. optimal) adversary with access to unbounded computational power and auxiliary information \cite{tschantz2020sok}. Erring on the side of security in this way is prudent, as it means that DP bounds always hold for weaker adversaries. However, the privacy guarantee of an algorithm under realistic conditions, where such adversaries may not exist, could be more optimistic than indicated. This naturally leads to the question what the \say{actual} privacy guarantees of algorithms are under relaxed adversarial assumptions.  

Works on empirical verification of DP guarantees \cite{carlini2022membership, jagielski2020auditing, nasr2021adversary} have recently led to two general findings:
\begin{enumerate}
    \item The DP guarantee in the worst case is (almost) tight, meaning that an improved analysis is not able to offer stronger bounds on existing algorithms under the same assumptions;
    \item A relaxation of the threat model on the other hand leads to dramatic improvements in the empirical DP guarantees of the algorithm.
\end{enumerate}

Motivated by these findings, we initiate an investigation into a minimal threat model relaxation which results in an \say{almost optimal} adversary. Complementing the aforementioned empirical works, which instantiate adversaries who conduct membership inference tests, we assume a formal viewpoint but retain the hypothesis testing framework. Our contributions can be summarised as follows:
\begin{itemize}
    \item We begin by introducing a mild formal relaxation of the usual DP assumption of a Neyman-Pearson-Optimal (NPO) adversary to a Generalised Likelihood Ratio Testing (GLRT) adversary. We discuss the operational significance of this formal relaxation in Section \ref{sec:background};
    \item In this setting, we provide tight privacy guarantees for the Gaussian mechanism in the spirit of Gaussian DP (GDP) and $(\varepsilon, \delta)$-DP, which we show to be considerably stronger than under the worst-case assumptions, especially in the high privacy regime.
    \item We provide composition results and subsampling guarantees for our bounds for use e.g. in deep learning applications.
    \item We find that --contrary to the worst-case setting-- the performance of the adversary in the GLRT relaxation is dependent on the dimensionality of the query, with high-dimensional queries having stronger privacy guarantees. We link this phenomenon to the asymptotic convergence of our bounds to an amplified GDP guarantee. 
    \item Finally, we experimentally evaluate our bounds, showing them to be tight against empirical adversaries. 
\end{itemize}

\section{Prior Work}
\textbf{Empirical verification of DP}: Several prior works have investigated DP guarantees from an empirical point-of-view. For instance, \cite{jagielski2020auditing} utilised data poisoning attacks to verify the privacy guarantees of DP-SGD, while \cite{nasr2021adversary} \textit{instantiate} adversaries in a variety of settings and test their membership inference capabilities. A similar work in this spirit is \cite{humphries2020differentially}.  

\textbf{Formalisation of membership inference attacks}: \cite{shokri2017membership} is among the earliest works to formalise the notion of a membership inference attack against a machine learning model albeit in a \textit{black-box} setting, where the adversary only has access to predictions from a targeted machine learning model. Follow-up works like \cite{ye2021enhanced, carlini2022membership} have extended the attack framework to a variety of settings. Recent works by \cite{sablayrolles2019white} or by \cite{mahloujifar2022optimal} have also provided formal bounds on membership inference success in a DP setting.

\textbf{Software tools and empirical mitigation strategies}: Alongside the aforementioned works, a variety of software tools has been proposed to \textit{audit} the privacy guarantees of ML systems, such as \textit{ML-Doctor} \cite{liu2022ml} or \textit{ML Privacy Meter} \cite{ye2021enhanced}. Such tools operate on the premises related to the aforementioned \textit{adversary instantiation}.

Of note, DP is not the only technique to defend against membership inference attacks (although it is among the few formal ones). Works like \cite{liu2021generalization, usynin2022zen} have proposed so-called \textit{model adaptation} strategies, that is, methods which empirically harden the model against attacks without necessarily offering formal guarantees.

\textbf{Gaussian DP, numerical composition and subsampling amplification}: Our analysis relies heavily on the hypothesis testing interpretation of DP and specifically Gaussian DP (GDP) \cite{dong2021gaussian}, however we present our privacy bounds in terms of the more familiar Receiver-Operator-Characteristic (ROC) curve similarly to \cite{Kaissis_Knolle_Jungmann_Ziller_Usynin_Rueckert_2022}. We note that for the purposes of the current work, the guarantees are identical. Some of our guarantees have no tractable analytic form, instead requiring numerical computations, similar to \cite{gopi2021numerical, zhu2022optimal}. We make strong use of the duality between GDP and \textit{privacy profiles} for privacy amplification by subsampling, a technique described in \cite{balle2020privacy}.

\section{Background}\label{sec:background}
\subsection{The DP threat model}
We begin by briefly formulating the DP threat model in terms of an abstract, non-cooperative \textit{membership inference game}. This will then allow us to relax this threat model and thus present our main results in a more comprehensible way. Throughout, we assume two parties, a \textit{curator} $\mathcal{C}$ and an \textit{adversary} $\mathcal{A}$ and will limit our purview to the Gaussian mechanism of DP.

\begin{definition}[DP membership inference game]
Under the DP threat model, the game can be reduced to the following specifications. We note that any added complexity beyond the level described below can only serve to make the game harder for $\mathcal{A}$ and thus improve privacy.
\begin{enumerate}
    \item The adversary $\mathcal{A}$ selects a function $f: \mathcal{X} \rightarrow \mathbb{R}^n$ where $\mathcal{X}$ is the space of datasets with (known) global $\ell_2$-sensitivity $\Delta$ and crafts two adjacent datasets $D$ and $D'$ such that $D \coloneqq \lbrace A \rbrace$ and $D' \coloneqq \lbrace A, B \rbrace$. Here, $A,B$ are the data of two individuals and \textbf{fully known} to $\mathcal{A}$. We denote the adjacency relationship by $\simeq$. 
    \item The curator $\mathcal{C}$ secretly evaluates either $f(D)$ or $f(D')$ and publishes the result $y$ with Gaussian noise of variance $\sigma^2\mathbf{I}^n$ calibrated to $\Delta$.
    \item The adversary $\mathcal{A}$, using all available information, determines whether $D$ or $D'$ was used for computing $y$.
\end{enumerate}
The game is considered won by the adversary if they make a correct determination.
\end{definition}
Under this threat model, the process of computing the result and releasing it with Gaussian noise is the DP mechanism. Note that the aforementioned problem can be reformulated as the problem of detecting the presence of a single individual given the output. This gives rise to the description typically associated with DP guarantees: \say{DP guarantees hold even if the adversary has access to the data of all individuals except the one being tested}. The reason for this is that, due to their knowledge of the data and the function $f$, $\mathcal{A}$ can always \say{shift} the problem so that (WLOG) $f(A) = 0$, from which it follows that $f(B) = \Delta$ (where the strict equality is due to the presence of only two points in the dataset and consistent with the DP guarantee).

More formally, the problem can thus be expressed as the following one-sided hypothesis test:
\begin{equation}
    \mathcal{H}_0: y = Z \;\; \text{vs.} \;\; \mathcal{H}_1: y = \Delta + Z, Z \sim \mathcal{N}(0, \sigma^2)
\end{equation}
and is equivalent to asking $\mathcal{A}$ to distinguish the distributions $\mathcal{N}(0, \sigma^2)$ and $\mathcal{N}(\Delta, \sigma^2)$ based on a single draw. The full knowledge of the two distributions' parameters renders both hypotheses \textit{simple}. In other words, $\mathcal{A}$ is able to compute the following log-likelihood ratio test statistic:
\begin{equation}\label{NPO_LR}
    \log \left(\frac{\text{Pr}(y \mid \mathcal{N}(\Delta, \sigma^2))}{\text{Pr}(y \mid \mathcal{N}(0, \sigma^2))} \right) = \frac{1}{2\sigma^2}\left(\vert y \vert^2 - \vert y- \Delta \vert^2\right),
\end{equation}
which depends only on known quantities. We call this type of adversary \textit{Neyman-Pearson-Optimal} (NPO) as they are able to detect the presence of the individual in question with the best possible trade-off between Type I and Type II errors, consistent with the guarantee of the Neyman-Pearson lemma \cite{neyman1933ix}. As is evident from Equation \eqref{NPO_LR}, the capabilities of an NPO adversary are independent of query dimensionality. Due to the isotropic properties of the Gaussian mechanism, the ability to form the full likelihood ratio allow $\mathcal{A}$ to \say{rotate the problem} in a way that allows them to linearly classify the output, which amounts to computing the multivariate version of the $z$-test. We remark in passing that this property forms the basis of \textit{linear discriminant analysis}, a classification technique reliant upon the aforementioned property. GDP utilises the worst-case capabilities of an NPO adversary as the basis for formulating a DP guarantee:
\begin{definition}[$\mu$-GDP, \cite{dong2021gaussian}]
A mechanism $\mathcal{M}$ preserves $\mu$-GDP if, $\forall D, D': D\simeq D'$, distinguishing between $f(D)$ and $f(D')$ based on an output of $\mathcal{M}$ is at least as hard (in terms of Type I and Type II error) as distinguishing between $\mathcal{N}(0, 1)$ and $\mathcal{N}(\mu, 1)$ with $\mu=\frac{\Delta}{\sigma}$.
\end{definition}
We stress that this guarantee is symmetric and given (1) over the randomness of $\mathcal{M}$, i.e. considering the datasets and $f$ as deterministic and (2) without consideration to (i.e. over) the adversary's prior. GDP utilises a \textit{trade-off} function/curve $T(x)$ to summarise the set of all achievable Type I and Type II errors by the adversary. An equivalent formulation is the \textit{testing region} in \cite{kairouz2015composition}. The \textit{trade-off} function is identical to the complement of the ROC curve (i.e. $T(x)=1-R(x)$), that is, the curve which plots the probability of correctly selecting $\mathcal{H}_1$ when it is true (equivalently, the probability of true positives, probability of detection ($P_d$) or sensitivity) against the probability of falsely selecting $\mathcal{H}_1$ when $\mathcal{H}_0$ is true (equivalently, the probability of false positive ($P_f$), probability of false alarm or $1-$specificity. Due to its greater familiarity, we will utilise $R(x)$ throughout as the guarantee is identical. For more details, we refer to \cite{dong2021gaussian, Kaissis_Knolle_Jungmann_Ziller_Usynin_Rueckert_2022}. 

\subsection{Our threat model relaxation}
As the introductory section above outlines, the NPO adversary is the most powerful adversary imaginable (i.e. \textit{omnipotent}). Mathematically, the \say{key to omnipotence} is the aforementioned ability to form the full likelihood ratio. This crucially relies on either knowledge of or control over the dataset and/or function. In many realistic settings however, this assumption may be too pessimistic. A few examples:
\begin{itemize}
    \item In federated learning, adversarial actors don't have access to other participants' datasets as they only witness the outputs of models trained on these datasets.
    \item In (e.g. medical) settings where data is generated and securely stored in a single institution, access to the full dataset is highly improbable.
    \item When neural networks are trained on private data with DP-SGD, the adversary has no control over the function and only receives the output of the computation.
\end{itemize}
The goal of our work is to provide a formal and tight privacy guarantee for such settings. Such a guarantee can then e.g. be parameterised by the probability that the NPO setting does not apply. For example, it can be stated that with probability $p$, the adversary has no access to the dataset while with probability $1-p$ they do, where $p$ is optimally very close to $1$. Such a \textit{flexibilisation} of the threat model allows stakeholders to make holistic and sound decisions on the choice of privacy parameters while maintaining the worst-case outlook offered by DP. Assuming the formal point-of-view once more, the key difference between the NPO setting and the examples above is that, in the latter setting, $\mathcal{A}$ is \textbf{unable to form the full likelihood ratio} as at least one parameter of the distribution under $\mathcal{H}_1$ is unknown. This renders $\mathcal{H}_1$ a \textit{composite} hypothesis. We define the threat model relaxation as follows:
\begin{definition}[Relaxed threat model]
We say that an adversary operates under a relaxed threat model if they are unable to fully specify the distributions of \emph{at least one of} $\mathcal{H}_0$ or $\mathcal{H}_1$ because some parameters are unknown to them and must thus be estimated.   
\end{definition}
To formally analyse this threat model in the rest of the paper, we will make the \textit{smallest possible} relaxation to the adversary and assume they are \textit{only lacking full knowledge of a single parameter} of the distribution of $\mathcal{H}_1$. This results in a \textit{nearly omnipotent} adversary. As is shown, even this minimal relaxation leads to a substantial amplification in privacy in certain regimes which are very relevant to everyday practice. For example, consider an adversary $\mathcal{A}_R$ who has full knowledge of $f$ and $\mathcal{M}$, but no access to $D$ or $D'$. Using this information, the adversary can infer all required information \textit{except the sign of $\Delta$}, as they now only know that $\Vert f(D) - f(D') \Vert_2 = \Delta$. Of note, we use the term \textit{sign} here to denote the multivariate \textit{signum} function $\operatorname{sgn}(x) = \frac{x}{\Vert x \Vert_2}$, i.e. the principal direction of the vector in $d$-dimensional space. More intuitively, $\mathcal{A}_R$ knows that $f(D)$ and $f(D')$ have \textit{distance} $\Delta$ from each-other, but not \textit{where} they are located with respect to one-another spatially. This situation results in the following hypothesis test (in the case of scalar $y$):
\begin{equation}
    \mathcal{H}_0: y = Z \;\; \text{vs.} \;\; \mathcal{H}_1: y = \pm \Delta + Z, Z \sim \mathcal{N}(0, \sigma^2)
\end{equation}
Evidently $\mathcal{H}_1$ is not unique and thus unrealisable. Formally, a \textit{uniformly most powerful} test does not exist to distinguish between the hypotheses. The adversary thus has to resort in a category of tests summarised under the term \textit{Generalised Likelihood Tests}, and we term such an adversary a \textit{Generalised Likelihood Test} (GLRT) adversary. For a general $\boldsymbol{y}$, the hypothesis test takes the following form:
\begin{equation}
    \mathcal{H}_0: \boldsymbol{y} \sim \mathcal{N}(\boldsymbol{0}, \sigma^2\mathbf{I}) \;\; \text{vs.} \;\; \mathcal{H}_1: \boldsymbol{y} \sim \mathcal{N}(\boldsymbol{\nu}, \sigma^2\mathbf{I}),
\end{equation}
where $\boldsymbol{\nu} \in \mathbb{R}^d$ is unknown. Alternatively, we can formulate the equivalent \textit{two-sided} test:
\begin{equation}
    \mathcal{H}_0: \boldsymbol{\nu} = \mathbf{0} \;\; \text{vs.} \;\; \mathcal{H}_1: \boldsymbol{\nu} \neq \mathbf{0}.
\end{equation}
Then, the adversary will reject $\mathcal{H}_0$ when the following ratio takes a large value:
\begin{equation}
    \frac{\text{Pr}(\boldsymbol{y} \mid \boldsymbol{\nu}^{\ast}, \mathcal{H}_1)}{\text{Pr}(\boldsymbol{y} \mid \mathcal{H}_0)},
\end{equation}
where $\boldsymbol{\nu}^{\ast}$ is the maximum likelihood estimate of $\boldsymbol{\nu}$. As seen in the proof to Theorem 1 below, in the relaxed threat model membership inference game, this will lead to classifying $\boldsymbol{y}$ by magnitude alone. We will show that this has two effects: 
\begin{enumerate}
    \item It substantially amplifies privacy for the individuals in $D$/$D'$;
    \item This amplification is moreover dependent on the dimensionality of the query with stronger privacy for higher-dimensional queries.
\end{enumerate}

\section{Results}

\subsection{ROC curve in the GLRT setting}
We are now ready to precisely describe the membership inference capabilities of the GLRT adversary in terms of their hypothesis testing prowess.

\begin{theorem}
Let $R(x)$ be the ROC curve of the following test:
\begin{equation}\label{test_1}
    \mathcal{H}_0: \boldsymbol{y} \sim \mathcal{N}(\boldsymbol{0}, \sigma^2\mathbf{I}) \;\; \text{vs.} \;\; \mathcal{H}_1: \boldsymbol{y} \sim \mathcal{N}(\boldsymbol{\nu}, \sigma^2\mathbf{I})
\end{equation}
and $R'(x)$ be the ROC curve of the following test:
\begin{equation}
    \mathcal{H}'_0: \boldsymbol{y} \sim \mathcal{N}(\boldsymbol{\nu}, \sigma^2\mathbf{I}) \;\; \text{vs.} \;\; \mathcal{H}'_1: \boldsymbol{y} \sim \mathcal{N}(\boldsymbol{0}, \sigma^2\mathbf{I}),
\end{equation}
where $y$ is the output of a Gaussian mechanism with noise variance $\sigma^2$ known to the adversary and $\boldsymbol{\nu}$ be an unknown vector of dimensionality $d$ with $\Vert \boldsymbol{\nu} \Vert_2 \leq \Delta$ for some constant $\Delta$ which is known to the adversary. Then, the following hold:
\begin{equation}\label{R}
    R(x) = \Psi_{\chi^2_d\left(\frac{\Delta^2}{\sigma^2}, \sigma^2 \right)} \left( \Psi^{-1}_{\chi^2_d\left(0, \sigma^2 \right)}(x)\right)
\end{equation}
and
\begin{equation}
    R'(x) = \Phi_{\chi^2_d\left(0, \sigma^2 \right)} \left(\Phi^{-1}_{\chi^2_d\left(\frac{\Delta^2}{\sigma^2}, \sigma^2 \right)}(x) \right).
\end{equation}
Above, $\Phi$ is the cumulative distribution function, $\Psi$ the survival function, $\Phi^{-1}$ and $\Psi^{-1}$ their respective inverses, $\chi^2_d(\lambda, \sigma^2)$ denotes the noncentral chi-squared distribution with $d$ degrees of freedom, noncentrality parameter $\lambda$ and scaling factor $\sigma^2$. Setting $\lambda=0$ recovers the central chi-squared distribution.
\end{theorem}

\begin{proof}
We give the proof for $R$ (Equations \eqref{test_1} and \eqref{R}), the proof for $R'$ follows from it. Since the adversary does not know the true value of $\boldsymbol{\nu}$, the full likelihood ratio cannot be formed and a Neyman-Pearson approach (i.e. $z$-test) is out of the question. Instead, classification is done by magnitude. As described above, the null hypothesis is rejected if the following ratio takes a large value:
\begin{equation}
    \frac{\text{Pr}(\boldsymbol{y} \mid \boldsymbol{\nu}^{\ast}, \mathcal{H}_1)}{\text{Pr}(\boldsymbol{y} \mid \mathcal{H}_0)}.
\end{equation}
However, the maximum likelihood estimate $\boldsymbol{\nu}^{\ast}$ is just $\boldsymbol{y}$, thus the likelihood ratio is simplified to (cancelling the normalisation factors):
\begin{equation}
    \frac{\exp\left\{-\frac{\lVert \boldsymbol{y} - \boldsymbol{\nu}\rVert_2^2}{2\sigma^2}\right\}}{\exp\left\{-\frac{\lVert \boldsymbol{y}\rVert_2^2}{2\sigma^2}\right\}} = \frac{1}{\exp\left\{-\frac{\lVert \boldsymbol{y}\rVert_2^2}{2\sigma^2}\right\}}.
\end{equation}
Taking logarithms and collecting known terms, we obtain that the test statistic is:
\begin{equation} \label{test_statistic}
    T(\boldsymbol{y}) = \lVert \boldsymbol{y} \rVert^2_2 \lessgtr c^2, 
\end{equation}
where $c^2$ is a (non-negative) cut-off value chosen to satisfy a desired level of Type I/Type II errors (or significance level and power level or $P_d$ and $P_f$).
Then, by the definition of the null and alternative hypotheses,
\begin{align}
    & P_f = \text{Pr}(\lVert \boldsymbol{y} \rVert^2 > c^2 \mid \mathcal{H}_0) \; \text{and} \\
    & P_d = \text{Pr}(\lVert \boldsymbol{y} \rVert^2 > c^2 \mid \mathcal{H}_1)
\end{align}
hold.
Under $\mathcal{H}_0$, $T(\boldsymbol{y})$ follows a central chi-squared distribution with $d$ degrees of freedom and scale $\sigma^2$, as it is the distribution of the magnitude of a draw from a $d$-dimensional multivariate Gaussian random variable with spherical covariance $\sigma^2\mathbf{I}$. Thus:
\begin{equation}
    P_f = \Psi_{\chi^2_d \left(0, \sigma^2\right)}(c^2) \Leftrightarrow c^2 = \Psi^{-1}_{\chi^2_d\left(0, \sigma^2 \right)}(P_f),
\end{equation}
Under the alternative hypothesis, the distribution is noncentral chi squared with $d$ degrees of freedom, scale $\sigma^2$ and noncentrality parameter $\frac{\Delta^2}{\sigma^2}$ as it is the distribution of the magnitude of a draw from a $d$-dimensional multivariate Gaussian random variable with spherical covariance $\sigma^2\mathbf{I}$. The ROC curve plots $P_f(c^2), P_d(c^2)$ parametrically for all values of $c^2$. Observing that in the plot, $P_f$ is the $x$-coordinate, we substitute the expression for $c^2$ from above and obtain:
\begin{equation}
    \Psi_{\chi^2_d\left(\frac{\Delta^2}{\sigma^2}, \sigma^2 \right)} \left( \Psi^{-1}_{\chi^2_d\left(0, \sigma^2 \right)}(x) \right).
\end{equation}
A concrete example for $d=1$ can be found in the proof to Corollary 1.
\end{proof}

Computing both $R(x)$ and $R'(x)$ is imposed by the requirement of DP to hold when $D$ and $D'$ are swapped. $\mathcal{H}_0$/$\mathcal{H}_1$ and $\mathcal{H}'_0$/$\mathcal{H}'_1$ are easily identified as the tests to distinguish $D \coloneqq \lbrace A \rbrace$ from $D' \coloneqq \lbrace A, B \rbrace$ and $D'$ from $D$, respectively.
In practice, we will always use the symmetrified and concavified version of the ROC curve $R_s(x)$, as prescribed in \cite{dong2021gaussian}. We refer to the aforementioned manuscript for details of the construction or $R_s(x)$. In brief, the curve is constructed by taking the concave upper envelope of $R$ and $R'$ by linearly interpolating between the points where the slope of the curves is parallel to the diagonal of the unit square. By symmetrifying the curves, it is guaranteed that the correct (worst-case) bound is used to convert to $(\varepsilon, \delta)$-DP by Legendre-Fenchel conjugation, a procedure we outline below. Examples of the curves in question are shown in Figure \ref{fig:roc_curves}.

\begin{figure}[ht]
    \centering
    \includegraphics[width=0.45\textwidth]{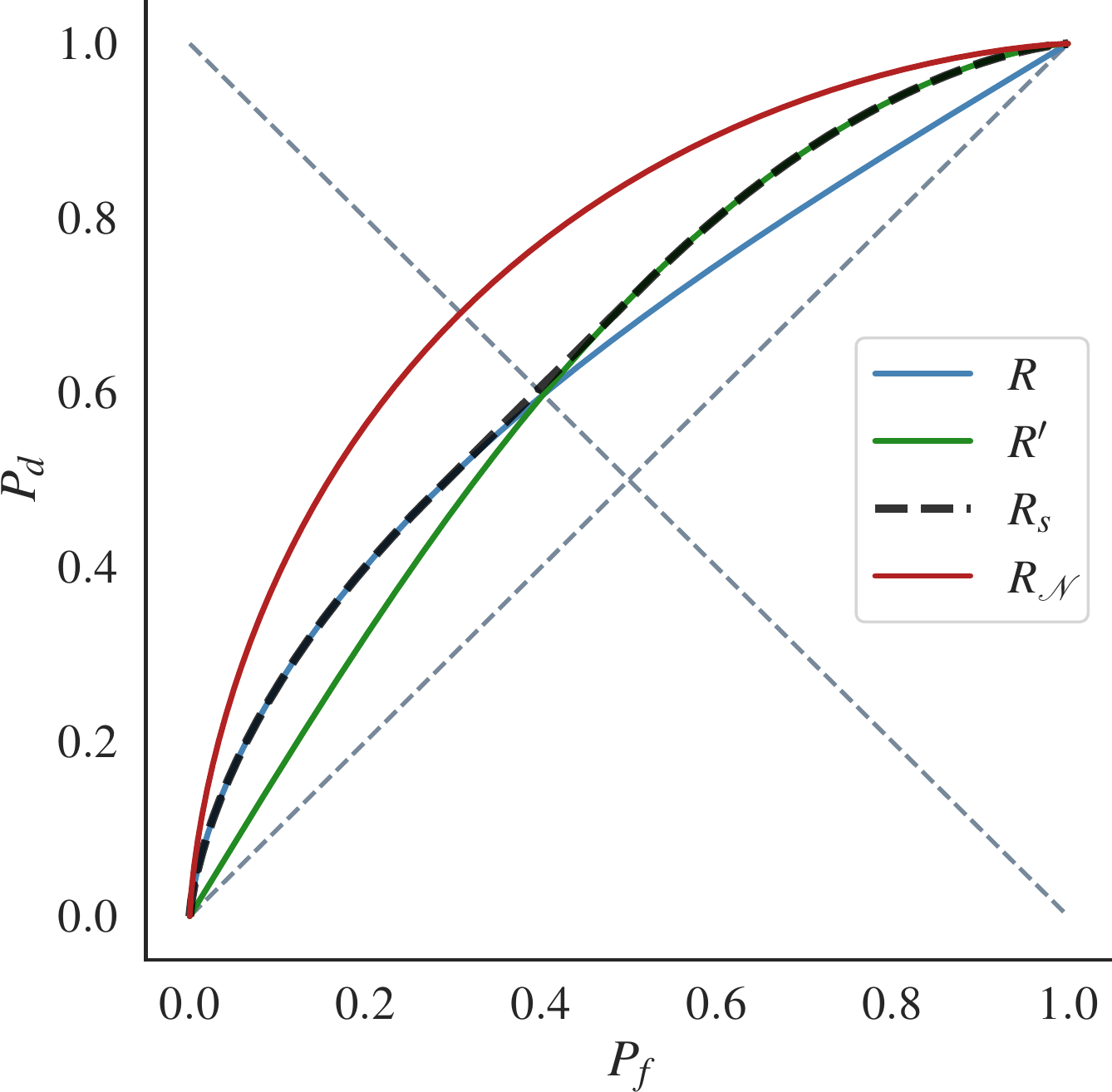}
    \caption{Exemplary visualisations of $R(x)$ (blue) and $R'(x)$ (green) as well as the concavified and symmetrified ROC curve $R_s(x)$ (dashed black) for a Gaussian mechanism with $\Delta=1, \sigma^2=1$ under the GLRT assumption. $R$ and $R'$ are symmetric with respect to reflection over the off-diagonal but not symmetric ROC curves, whereas $R_s(x)$ is a symmetric ROC curve. For reference, the ROC curve of an NPO adversary $R_{\mathcal{N}}$ at the same $\Delta$ and $\sigma^2$-values is shown in red. Observe that $R_{\mathcal{N}}$ is also a symmetric ROC curve.}
    \label{fig:roc_curves}
\end{figure}

We are particularly interested in the case where $d=1$, as it represents the worst-case scenario in terms of privacy (respectively the easiest problem for the adversary, see Section \ref{sec:dimensionality}). We emphasise that in this case, all the adversary is unaware of is the \textit{sign of the mean under the alternative hypothesis}.
\begin{corollary}
When $d=1$, $R(x)$ admits the following closed-form expression:
\begin{equation}
    Q\left(Q^{-1}\left(\frac{x}{2}\right)-\frac{\Delta}{\sigma}\right) + Q\left(Q^{-1}\left(\frac{x}{2}\right)+\frac{\Delta}{\sigma}\right),
\end{equation}
where $Q$ is the survival function of the standard normal distribution and $Q^{-1}$ its inverse. Compare the ROC curve of the NPO adversary shown in red in Figure \ref{fig:roc_curves}:
\begin{equation}\label{eq:gauss_roc}
    Q\left(Q^{-1}(x) - \frac{\Delta}{\sigma}\right).
\end{equation}
\end{corollary}

\begin{proof}
This is a special case of the proof to Theorem 1 above, but now for scalar $y$ since $d=1$. Under the null hypothesis, the survival function of the central chi-squared distribution with one degree of freedom admits an analytical form:
\begin{align} \label{erfc}
    1-\frac{\gamma\left( \frac{1}{2}, \frac{c^2}{2{\sigma^2}} \right)}{\Gamma\left( \frac{1}{2}\right)} = &1-\frac{\sqrt{\pi}\operatorname{erf}\left(\sqrt{\frac{c^2}{2{\sigma^2}}}\right)}{\sqrt{\pi}} = \\ = \,
    &\operatorname{erfc}\left(\frac{c}{\sqrt{2}\sigma}\right),
\end{align}
where $\gamma$ is the lower incomplete gamma function, $\Gamma$ the gamma function, we inserted $\frac{c^2}{\sigma^2}$ to account for the scale and $\operatorname{erf}, \operatorname{erfc}$ are the error function and complementary error function of the Gaussian distribution, respectively. We can now exploit the following pattern:
\begin{equation}
    Q(k) = \frac{1}{2}\operatorname{erfc}\left(\frac{k}{\sqrt{2}}\right),
\end{equation}
so the term in Equation \eqref{erfc} can be written as $2Q\left(\frac{c}{\sigma}\right)$. Since the $Q$ function is invertible, we have that $c=\sigma Q^{-1}\left(\frac{x}{2}\right)$.
Similarly, for the alternative hypothesis, we have:
\begin{align}
    &\mathsf{Q}_{M\frac{1}{2}}\left (\sqrt{\frac{\Delta^2}{\sigma^2}}, \sqrt{\frac{c^2}{\sigma^2}} \right) = \\ = \,
    & \mathsf{Q}_{M\frac{1}{2}}\left (\frac{\Delta}{\sigma}, \frac{c}{\sigma} \right),
\end{align}
where $\mathsf{Q}_{M\frac{1}{2}}$ is the Marcum Q-function of order $\frac{1}{2}$, which is the survival function of the noncentral chi-squared distribution with one degree of freedom, and we have substituted $\Delta^2$ as the noncentrality to account for the mean under the alternative hypothesis and divided by $\sigma^2$ to account for the scaling. Substituting the expression for $c$ from above, we obtain:
\begin{align}
    &\mathsf{Q}_{M\frac{1}{2}}\left (\frac{\Delta}{\sigma}, \frac{\sigma Q^{-1}\left(\frac{x}{2}\right)}{\sigma} \right) = \\ = \,
    &\mathsf{Q}_{M\frac{1}{2}}\left (\frac{\Delta}{\sigma}, Q^{-1}\left(\frac{x}{2}\right) \right). \label{marcum}
\end{align}
The Marcum Q-function of order $\frac{1}{2}$ also admits a closed form:
\begin{equation}\label{marcum_2}
    \mathsf{Q}_{M\frac{1}{2}}(a,b) = \frac{1}{2}\left(\operatorname{erfc}\left(\frac{b-a}{\sqrt{2}} \right) + \operatorname{erfc}\left(\frac{b+a}{\sqrt{2}} \right) \right).
\end{equation}
Using the pattern $\frac{1}{2}\operatorname{erfc}\left(\frac{k}{\sqrt{2}}\right)$ as $Q(k)$, we rewrite Equation \eqref{marcum_2} as $\mathsf{Q}_{M\frac{1}{2}}(a,b) = Q(b-a) + Q(a+b)$. Finally, we substitute the arguments from Equation \ref{marcum} and obtain:
\begin{equation}
    Q\left(Q^{-1}\left(\frac{x}{2}\right) - \frac{\Delta}{\sigma}\right) + Q\left(Q^{-1}\left(\frac{x}{2}\right) + \frac{\Delta}{\sigma}\right),
\end{equation}
which completes the proof. The proof of Equation \eqref{eq:gauss_roc} can be found in \cite{dong2021gaussian} or \cite{Kaissis_Knolle_Jungmann_Ziller_Usynin_Rueckert_2022} and follows by standard properties of the Gaussian survival and cumulative distribution functions.
\end{proof}

Unfortunately, $R'(x)$ has no closed-form expression for any $d$, as the noncentral chi-squared distribution cannot --in general-- be inverted analytically. However, it is easy to invert numerically as it is monotonic and routines to evaluate the inverse to high precision are available in all standard numerical software libraries and we describe a differentiable implementation below. Observe also that the ROC curve (for a given $d$) depend only on the ratio of $\Delta$ to $\sigma$. This \say{signal to noise}(SNR)-type argument is also made by \cite{Kaissis_Knolle_Jungmann_Ziller_Usynin_Rueckert_2022} for the NPO adversary.

\subsection{Composition}
The results presented above bound the capability of the GLRT adversary in the setting of a single query release. We are now interested in extending the guarantee to the composition setting, where the adversary receives the results of $N$ queries. We stress that we adopt a worst-case outlook on composition, i.e. we assume that in the setting of $N$-fold composition, the adversary incorporates all knowledge gained from queries $1, \dots, N-1$ to improve their probability of success at the membership inference game. This interpretation is consistent with the DP threat model.

\begin{theorem}
Let $\mathcal{M}$ be a Gaussian mechanism on a function with sensitivity $\Delta$, noise variance $\sigma^2$ and output dimensionality $d$. Then, under $N$-fold homogeneous composition, the ROC curves are given by:
\begin{equation}
    R(x)^{\otimes N} = \Psi_{\chi^2_d(\lambda_{\text{comp}}, \sigma^2_{\text{comp}})} \left( \Psi^{-1}_{\chi^2_d\left(0, \sigma^2_{\text{comp}} \right)}(x)\right)
\end{equation}
and
\begin{equation}
    R'(x)^{\otimes N} = \Phi_{\chi^2_d\left(0, \sigma^2_{\text{comp}} \right)} \left(\Phi^{-1}_{\chi^2_d\left(\lambda_{\text{comp}}, \sigma^2_{\text{comp}} \right)}(x) \right).
\end{equation}
with 
\begin{equation}
    \lambda_{\text{comp}\,\otimes N} = \frac{N\Delta^2}{\sigma^2} \;\; \text{and} \;\; \sigma^2_{\text{comp}\,\otimes N} = \frac{\sigma^2}{N}.
\end{equation}
Here, $\otimes N$ denotes $N$-fold composition.
\end{theorem}
\begin{proof}
Since we allow the adversary to collect all $N$ samples before having to commit to one of the two hypotheses, they can exploit the isotropic property of Gaussian noise to \say{average out} the noise, so that the test statistic in Equation \eqref{test_statistic} becomes:
\begin{equation}
    \left \Vert \frac{1}{N} \sum_{i=1}^N \boldsymbol{y}_i \right \Vert_2^2 \lessgtr c^2,
\end{equation}
This transforms the distributions of the hypotheses to:
\begin{align}
    &\mathcal{H}_0: \boldsymbol{y} \sim \mathcal{N}\left(\boldsymbol{0}, \frac{\sigma^2}{N}\mathbf{I}\right) \;\; \text{vs.} \\ &\mathcal{H}_1: \boldsymbol{y} \sim \mathcal{N}\left(\boldsymbol{\nu}, \frac{\sigma^2}{N}\mathbf{I}\right).
\end{align}
The claim then follows by substituting $\sigma^2 \leftarrow \frac{\sigma^2}{N}$ in Equation \eqref{R}. 
\end{proof}

This result is an effect of the fact that under composition, Gaussian mechanisms remain Gaussian and (due to the independence and identical distribution), their magnitudes converge exactly to chi-squared distributions. In the setting of heterogeneous query dimensionality, sensitivity or noise, either the ROC or the privacy profile (see below) can be composed numerically using the characteristic function representation of \cite{zhu2022optimal} (i.e. either the Fourier integral or the Fourier transform, respectively), including subsampling amplification. In the setting of high-dimensional queries composed homogeneously over many rounds, such as in deep learning applications, one can leverage the asymptotic results shown below.

\subsection{Query dimensionality} \label{sec:dimensionality}
Interestingly, in the GLRT setting, query dimensionality has an adverse effect on classification capability (contrary to the NPO setting): It is substantially harder to detect the presence of an individual in a high-dimensional query output compared to a low-dimensional one (assuming the sensitivity and noise are identical). Formally, this can easily be verified by observing the monotonicity of the noncentral chi-squared cumulative distribution function/ quantile function under an increase in degrees of freedom (at a fixed noncentrality). 

Intuitively, the phenomenon can be understood by the effect of dimensionality on vector magnitude. In the setting of low query sensitivity, high noise and few observations, (that is to say, the setting we are interested in in privacy), the noise in high-dimensional space dominates the signal originating from the spatial separation of the vectors \cite{Urkowitz}. This phenomenon is equivalent to the \say{curse of dimensionality} observed in other fields of statistics and machine learning. Thus, an adversary who is in control of the function or can influence the training process will want to minimise the number of free parameters to detect the change in the output induced by the presence of a single individual with the highest possible sensitivity. This is e.g. the strategy employed by \cite{nasr2021adversary} and is the reason we focus on $d=1$ as the worst-case ROC curve above. Examples for the effect of dimensionality can be found in Figures \ref{fig:dof} and \ref{fig:dimension_epsilon} and in the experimental section below.

\begin{figure}[h]
    \centering
    \includegraphics[width=0.45\textwidth]{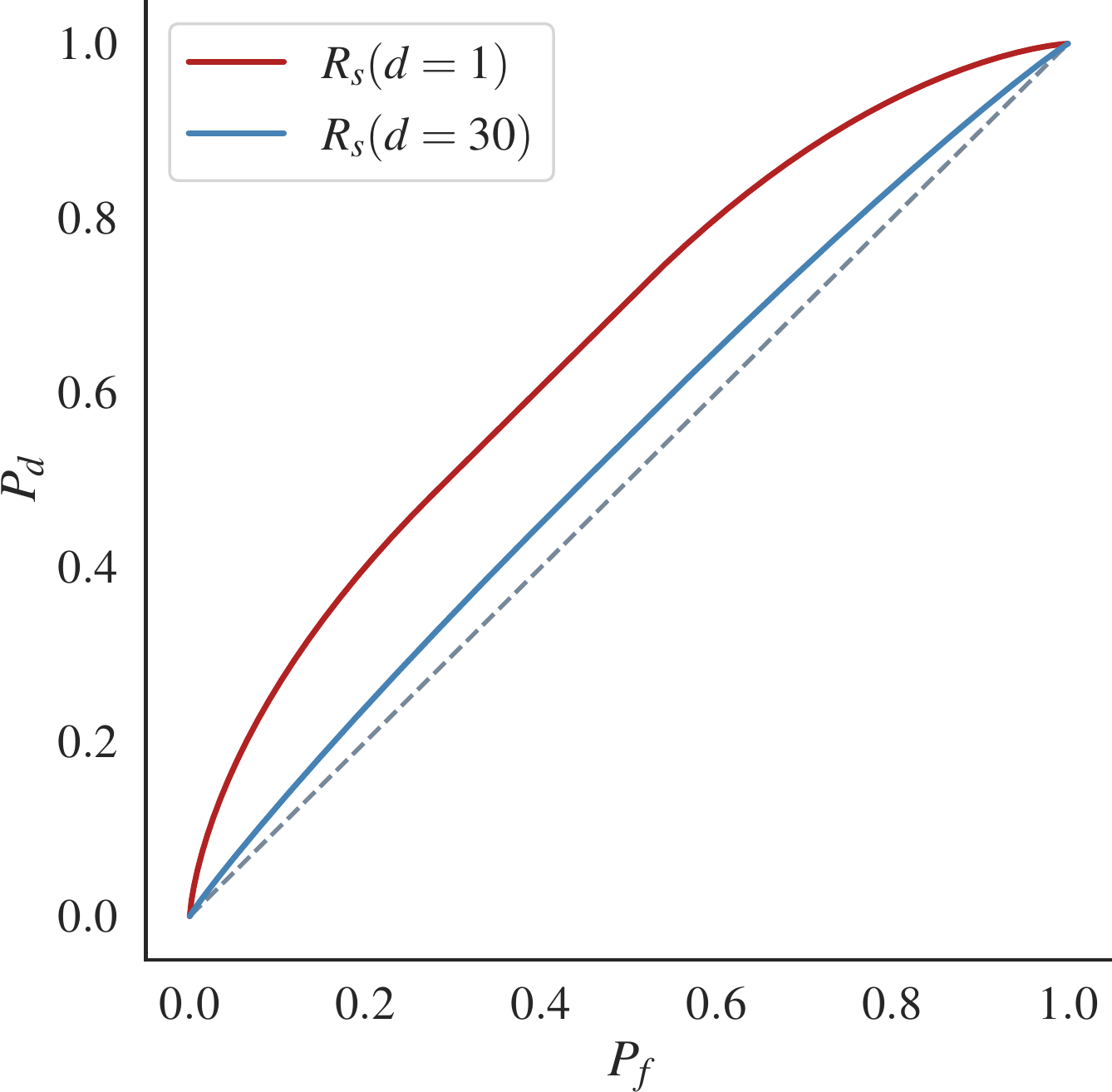}
    \caption{Effect on query dimensionality on the ROC curve. At a fixed $\Delta = \sigma^2 = 1$, the adversary classifying a query of dimensionality $d=1$ (red curve, $(\varepsilon=3.11, \delta=10^{-4})$-DP) will have substantially greater success than when the query has dimensionality $d=30$ (blue curve, $(0.46, 10^{-4})$-DP).}
    \label{fig:dof}
\end{figure}

\begin{figure}[h]
    \centering
    \includegraphics[width=0.45\textwidth]{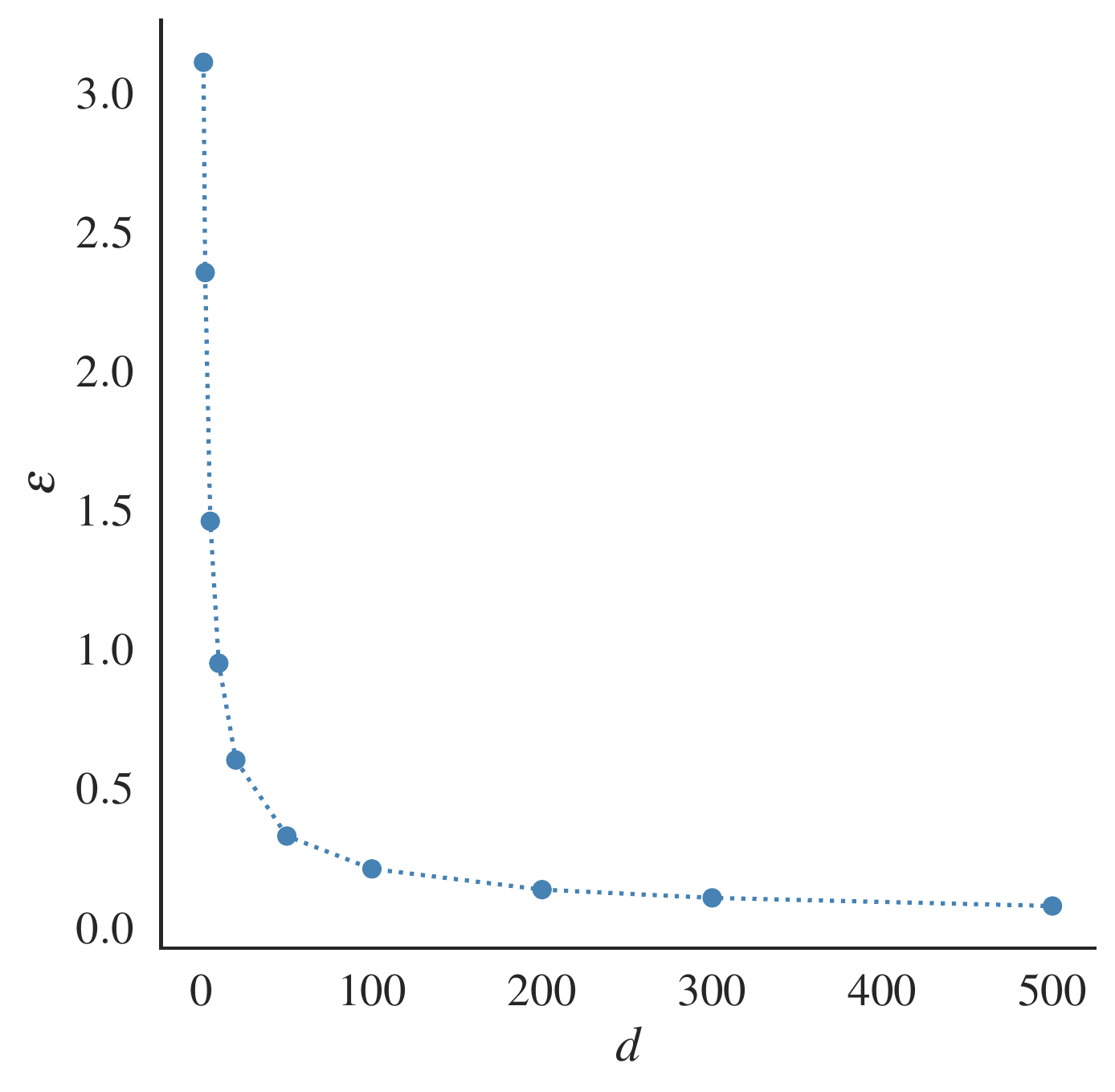}
    \caption{Query dimensionality plotted against the effective $\varepsilon$ value at $\delta=10^{-1}$ for a GLRT adversary. $\Delta=\sigma^2=1$, $d \in [1, 2, 5, 10, 20, 50, 100, 200, 300, 500]$. A marked decrease in effective $\varepsilon$ (conversion discussed in Section \ref{sec:eps_delta_conv}) is observed with increasing dimensionality.}
    \label{fig:dimension_epsilon}
\end{figure}

\subsection{Asymptotics}
To formalise the aforementioned result and foreshadow the upcoming findings that the GLRT relaxation is especially beneficial when the query dimensionality is large, the query sensitivity is low and noise magnitudes are high, we now consider the asymptotic behaviour of the ROC curves in this setting. Our results harken back to the observations of \cite{dong2021gaussian}, who --stated informally-- find that \say{very private mechanisms composed over many rounds asymptotically converge to GDP}. In practice, this means that, when the ROC curve is \say{very close} to the diagonal, i.e. for mechanisms with strong privacy guarantees, the resulting ROC curve under many rounds of composition should approach some version of the NPO Gaussian mechanism ROC curve shown in Equation \eqref{eq:gauss_roc}. Continuing the \say{signal-to-noise} argument from above: when many samples of a low SNR process are observed, the central limit theorem (CLT) applies and the system can be described well in terms of Gaussian distributions. We have the following formal result:
\begin{theorem}
Let $\mathcal{M}$ be a Gaussian mechanism on a function with sensitivity $\Delta$, noise variance $\sigma^2$ and output dimensionality $d \gg 1$ such that $\frac{\Delta}{\sigma} \ll 1$. As $N \rightarrow \infty$, the ROC curves of $\mathcal{M}$ under the GLRT assumption converge to:
\begin{align}
    &Q\left(\frac{1}{\sqrt{\frac{2 \Delta^{2} N}{d \sigma^{2}} + 1}}Q^{-1}(x) - \frac{\sqrt{2} \Delta^{2} N}{2 d \sigma^{2} \sqrt{\frac{2 \Delta^{2} N}{d \sigma^{2}} + 1}} \right) \label{eq:clt_1}
    \\ 
    &\approx \; Q\left(Q^{-1}(x) - \frac{N\Delta^2}{\sigma^2}\sqrt{\frac{1}{2d}} \right). \label{eq:clt_2}
\end{align}
The similarity of Equation \eqref{eq:clt_2} to Equation \eqref{eq:gauss_roc} thus leads us to conclude that $\mathcal{M}$ converges to $\mu$-GDP mechanism with:
\begin{equation}
    \mu = \frac{N\Delta^2}{\sigma^2}\sqrt{\frac{1}{2d}} \label{eq:mu}
\end{equation}
\end{theorem}
\begin{proof}
We will use the facts \cite{johnson1995continuous} that the mean of the central chi-squared distribution with $d$ degrees of freedom is $d$ and its variance is $2d$. The noncentral chi-squared distribution with $d$ degrees of freedom and noncentrality $\lambda$ has mean $d+\lambda$ and variance $2d+4\lambda$. Under the CLT, we thus have convergence in distribution as follows. Letting $\frac{\sigma^2}{N} \coloneqq \beta$:
\begin{align}
    &\chi^2_d(0, \beta) \rightarrow \mathcal{N}(\beta d, 2\beta^2d) \; \text{and} \\ 
    &\chi^2_d(\lambda, \beta) \rightarrow \mathcal{N}(\beta(d+\lambda), \beta^2(2d+4\lambda)).
\end{align}
We can thus use Equation \eqref{eq:gauss_roc} to derive the $P_f$ and $P_d$ similar to Theorem 1. We have:
\begin{equation}
    P_f = Q\left(\frac{c-\beta d}{\beta \sqrt{2d}} \right) \Rightarrow c = Q^{-1}(x) \beta \sqrt{2d} + \beta d
\end{equation}
and
\begin{equation}
    P_d = Q\left(\frac{c - \beta(d+\lambda)}{\beta \sqrt{2d+4\lambda}} \right).
\end{equation}
For the ROC curve, we thus obtain:
\begin{align}
    R(x) = &Q\left(\frac{Q^{-1}(x) \beta \sqrt{2d} + \beta d - \beta d - \beta \lambda}{\beta \sqrt{2d+4\lambda}} \right) = \\ = \,
    &Q\left( \frac{Q^{-1}(x)\sqrt{2d}-\lambda}{\beta \sqrt{2d+4\lambda}} \right) = \\ = \,
    &Q\left( \frac{Q^{-1}(x)-\sqrt{\frac{d}{2}}\frac{\lambda}{d}}{\sqrt{1+\frac{2\lambda}{d}}} \right).\label{fourty_two}
\end{align}
Substituting $\lambda \leftarrow \frac{N\Delta^2}{\sigma^2}$ and separating the terms, we obtain:
\begin{equation}
    Q\left(\frac{1}{\sqrt{\frac{2 \Delta^{2} N}{d \sigma^{2}} + 1}}Q^{-1}(x) - \frac{\sqrt{2} \Delta^{2} N}{2 d \sigma^{2} \sqrt{\frac{2 \Delta^{2} N}{d \sigma^{2}} + 1}} \right),
\end{equation}
which is the desired result in Equation \eqref{eq:clt_1}. To obtain Equation \eqref{eq:clt_2}, we further massage Equation \eqref{fourty_two}. Concretely, we let $\frac{\lambda}{d} \coloneqq \psi$ and Taylor expand the equation around $\psi=0$ to obtain:
\begin{equation}
    Q\left( Q^{-1}(x) - \psi \left( \sqrt{\frac{d}{2}} + Q^{-1}(x) \right) + \mathcal{O}(\psi^2) \right).
\end{equation}
When $d$ is large, $\mathcal{O}(\psi^2)$ vanishes and the $\sqrt{\frac{d}{2}}$ dominates the term in the parentheses, yielding:
\begin{equation}
    Q\left( Q^{-1}(x) - \psi \sqrt{\frac{d}{2}} \right).
\end{equation}
Finally, substituting $\psi \leftarrow \frac{\lambda}{d}$ and $\lambda \leftarrow \frac{N\Delta^2}{\sigma^2}$, we get:
\begin{equation}
    Q\left(Q^{-1}(x) - \frac{N\Delta^2}{\sigma^2}\sqrt{\frac{1}{2d}} \right),
\end{equation}
which is Equation \eqref{eq:clt_2}.
\end{proof}

Figure \ref{fig:clt} shows the accuracy of this result, which is in no small part due to the fact that both the central and the noncentral chi-squared distributions converge \textit{exactly} to Gaussian distributions when their degrees of freedom are high. Compared to Figure \ref{fig:roc_curves}, we also observe that the symmetrisation and concavification of $R$ and $R'$ are superfluous in this setting, as the curves become symmetric.

We remark that --although the CLT is a legitimate method to analyse composition-- the results are only valid asymptotically (i.e. in the infinite sample regime). A significantly faster (finite sample) convergence can be shown using the Berry-Eseen-theorem and the Edgeworth approximation implemented in \cite{wang2022analytical}, which can be used for composition. Moreover, we see that $\mu$ in Equation \eqref{eq:mu} is quadratic in the SNR (i.e. in $\frac{\Delta}{\sigma}$) and inversely proportional to query dimensionality, which explains why the GLRT relaxation is especially powerful in the high-privacy and high-dimensional regime (such as deep learning).

\begin{figure}[h]
    \centering
    \includegraphics[width=0.45\textwidth]{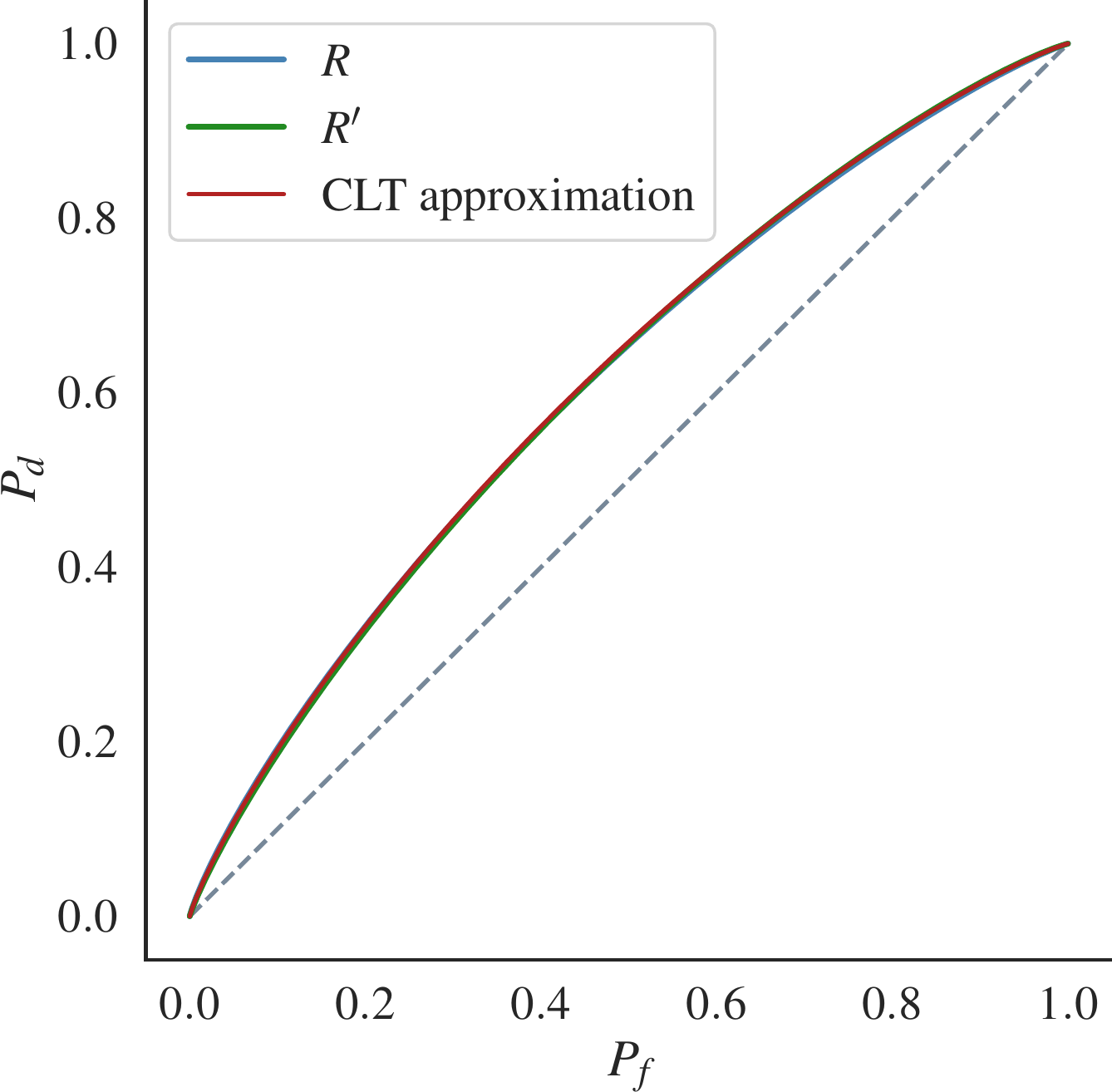}
    \caption{Convergence of $R$ (blue curve) and $R'$ (green curve) to the ROC curve of a $\mu$-GDP mechanism as shown in Equation \eqref{eq:clt_2}. In this case, $\Delta=1, \sigma^2=100, d=300$ and $N=1000$.}
    \label{fig:clt}
\end{figure}

\subsection{Conversion to $(\varepsilon, \delta)$-DP}\label{sec:eps_delta_conv}
As seen in Figure \ref{fig:roc_curves}, the GLRT adversary has diminished membership inference capabilities compared to the NPO adversary. The GDP framework has the benefit of encapsulating \textit{all} $P_d$/$P_f$ in one curve, but practitioners are often interested in the $(\varepsilon, \delta)$-guarantee, which --for many-- is considered \say{canonical}. Deriving this guarantee directly in the GLRT setting is difficult, as the specification of the privacy loss random variable involves terms which are not analytically tractable. Instead, we aim to exploit the lossless conversion between the ROC curve and the privacy profile, that is, an infinite collection of $(\varepsilon, \delta(\varepsilon))$-guarantees \cite{balle2020privacy}, which can be realised through Legendre-Fenchel duality (i.e. the concave conjugate (dual) of the ROC curve (or convex conjugate of the trade-off curve)). Our strategy is as follows:
\begin{enumerate}
    \item For a given $\Delta$ and $\sigma$, instantiate $R(x)$ and $R'(x)$;
    \item Compute the concavification and symmetrisation of the curves $R_s(x)$ as in \cite{dong2021gaussian};
    \item Compute the Legendre-Fenchel conjugate $R^{\ast}(x)$;
    \item Perform a change of variables from $(x, y(x))$ to $(\varepsilon, \delta(\varepsilon))$.
\end{enumerate}

The construction of $R^{\ast}(x)$ has to be performed numerically and has a success guarantee as the curve is monotonic. However, the \say{interesting} parts of the curve are the extreme locations where $P_d$ or $P_f$ are very small/large and thus the slope is very steep or near zero. These correspond to the low values of $\delta$ required for real-life applications. Since we found the technique to obtain the conjugate used in the \textrm{autodp} software package by \cite{zhu2022optimal} to be slow and sometimes numerically unstable, we implement the conjugate using double floating point precision central-difference numerical differentiation. We programmatically transform a subroutine $p$, which computes the value of a probability function of interest (in our case, $\Psi_{\chi^2}$ and $\Phi_{\chi^2}$) into a new subroutine $p'$ which computes the value of the derivative at the same point. Similarly, we derive $\Psi^{-1}_{\chi^2}$ and $\Phi^{-1}_{\chi^2}$. Using the aforementioned technique, we propose a refined procedure to compute $R^{\ast}(x)$. Our method is exemplified in Algorithm \eqref{alg:our_inversion}.

\begin{algorithm}[h]
\begin{algorithmic}
\Require Numerical ROC curve subroutines $R, R'$, symmetrisation and concavification subroutine \textbf{symm} described in \cite{dong2021gaussian}, numerical differentiation subroutine $\nabla$, root finding subroutine \textbf{root}, desired value of privacy parameter $\varepsilon$.
\Procedure{symm\_r}{$R, R'$}\Comment{Compute $R_s$ from $R, R'$}
    \State $R_s \gets \textbf{symm}(R, R')$ 
    \State \textbf{return} subroutine $R_s$
\EndProcedure

\Procedure{diff\_r}{$R_s$}\Comment{Compute $\mathcal{D}R_s$ from $R_s$}
    \State $\mathcal{D}R_s \gets \nabla R_s$
    \State \textbf{return} subroutine $\mathcal{D}R_s$
\EndProcedure

\Procedure{delta\_epsilon\_r}{$\varepsilon$}
    \State $m \gets e^{\varepsilon}$ 
    \State $x \gets$ \textbf{root} $R_s(x)=m$
    \State $y \gets \textbf{call} \; R_s(x)$ 
    \State $b \gets \textbf{call} \; \mathcal{D}R_s(x)$
    \State $\delta \gets -x*b+y$
    \State \textbf{return} $\delta$
\EndProcedure
\end{algorithmic}
\caption{Legendre-Fenchel conjugate computation using derivatives.}
\label{alg:our_inversion}
\end{algorithm}

The resulting privacy profile for the GLRT and NPO adversaries is exemplified in Figure \ref{fig:priv_profile}. Expectedly, the GLRT relaxation leads to a considerably lower value of $\delta$ for a given $\varepsilon$, especially in the high privacy regime $\varepsilon < 2$.  

\begin{figure}[h]
    \centering
    \includegraphics[width=0.45\textwidth]{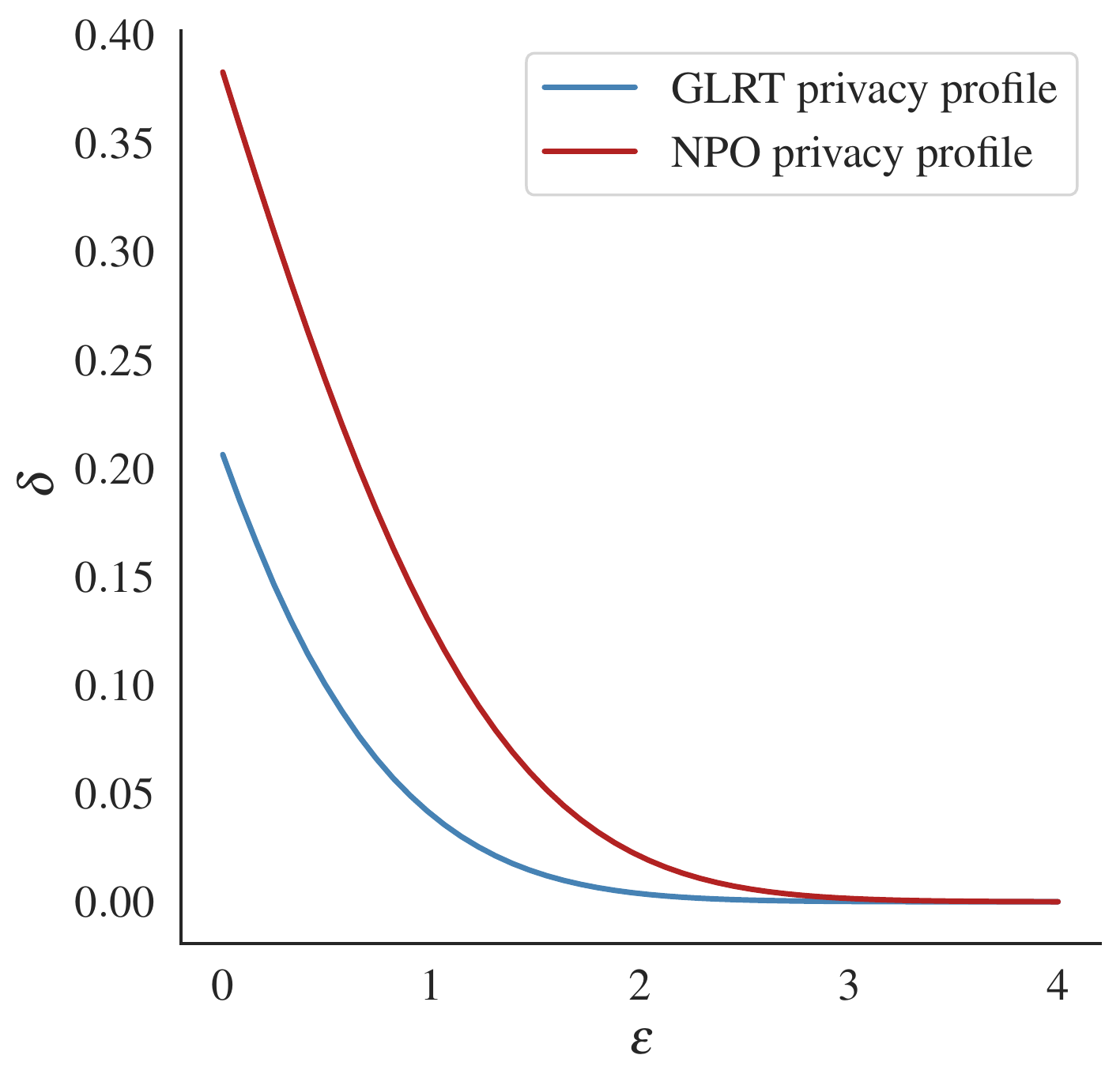}
    \caption{Privacy profile curves for the NPO (red curve) and the GLRT (blue curve) adversaries. $\Delta = \sigma^2 = 1$. A substantial reduction is observed in $\delta$, especially in the high privacy regime. The NPO privacy profile is computed as described in \cite{balle2020privacy}.}
    \label{fig:priv_profile}
\end{figure}

\subsection{Subsampling amplification}
DP applications can profit from the \textit{privacy amplification by sub-sampling} property which states that when secret subsamples are drawn from the database, the non-inclusion of a proportion of the database in the query results proportionally amplifies their privacy guarantees. This property also holds for the GLRT relaxation. Specifically, we analysed the privacy amplification of add/remove Poisson sampling as discussed in \cite{balle2020privacy} which states that, if a mechanism is $(\varepsilon, \delta)$-DP, the subsampled version of the mechanism with sampling probability $\gamma$ satisfies $\log(1+\gamma(e^{\varepsilon}-1)), \gamma \delta$-DP. To compute the guarantee in practice, we thus instantiate $R(x)$ and $R'(x)$, compute $R_s(x)$, then use duality to convert to $R^{\ast}(x)$ and finally amplify the resulting privacy profile. Since the privacy profile encapsulates \textit{all} $(\varepsilon, \delta(\varepsilon))$-pairs, it can be composed losslessly (including the subsampling analysis) by taking the Fourier transform, composing the characteristic functions and then taking the inverse of the resulting function as described in \cite{zhu2022optimal}. Figure \ref{fig:subsampling} exemplifies the amplification of the GLRT relaxation's privacy profile in comparison to the sub-sampled version of the NPO privacy profile.

\begin{figure}[h]
    \centering
    \includegraphics[width=0.45\textwidth]{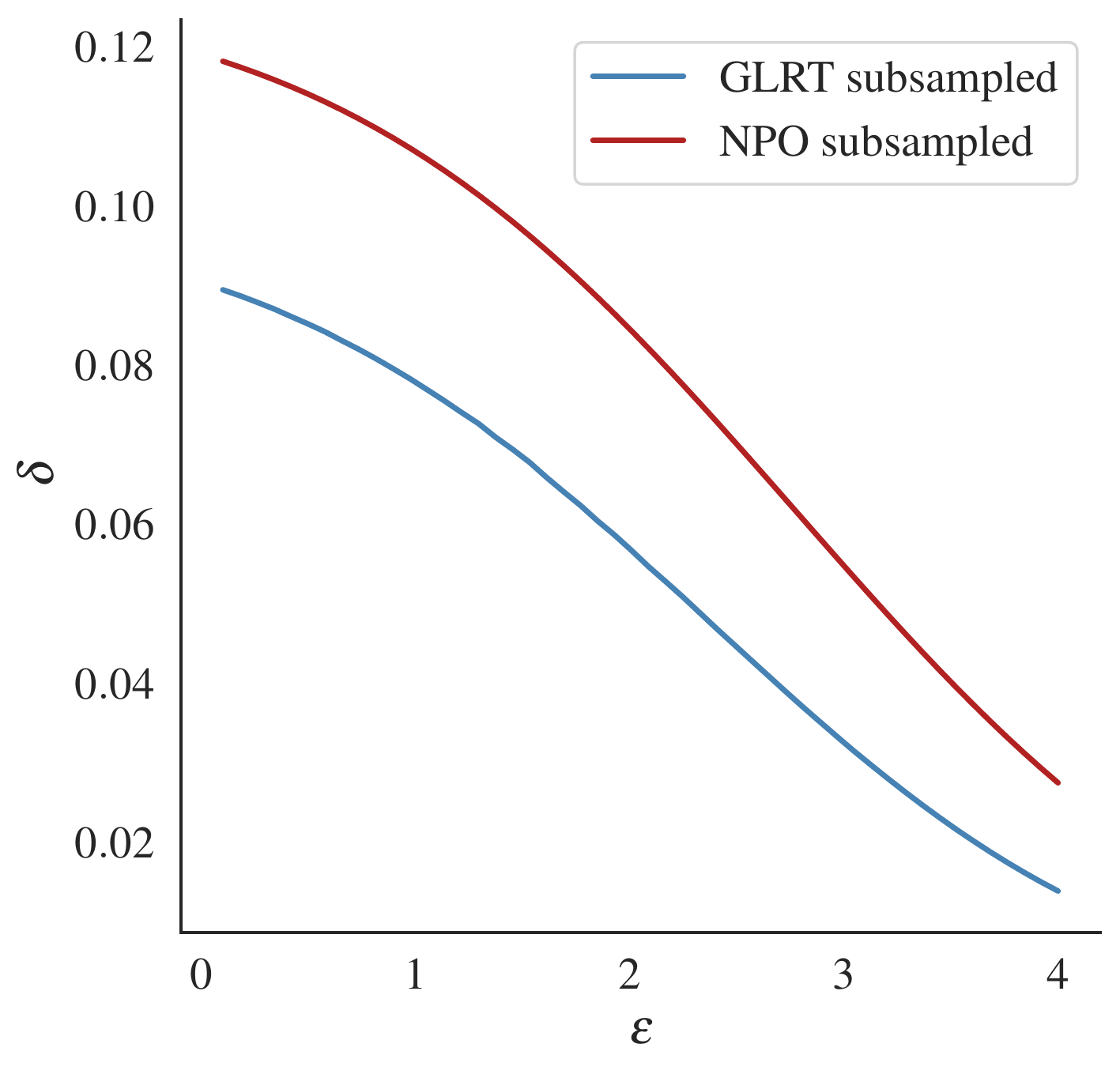}
    \caption{Subsampling amplification of the GLRT privacy profile at $\Delta=1, \sigma^2=36, \gamma=0.2$ computed over $N=100$ compositions, showing an improvement in $\delta$-to-$\varepsilon$ compared to the NPO privacy profile with the same parameters.}
    \label{fig:subsampling}
\end{figure}

\section{Numerical Experiments}
We conclude our investigation with a set of numerical experiments to demonstrate the tightness of our results in practice.

\subsection{Worst-case adversarial performance}
We begin by instantiating a worst-case (in terms of privacy) GLRT membership inference game. The adversary is faced with a pair of adjacent datasets containing at most two individuals and can control the query function so that its output on one of the individuals is $0$. Moreover, the query function is scalar-valued and has known global $L_2$-sensitivity $\Delta$. In the experiment below, $\Delta=1$. The noise added is of known magnitude $\sigma^2=36$. The adversary interacts with the system for $N=70$ composition rounds. We only require that the the adversary classifies the output of the final composition round correctly and they may use all information from the previous rounds. The adversarial task is thus:

\noindent
\textit{Classify an output $y=f(X)$ as coming from $\mathcal{N}(0, \sigma^2)$ or from $\mathcal{N}(\nu, \sigma^2)$ where $X$ is unknown and one of $D \coloneqq \lbrace A \rbrace$ or $D' \coloneqq \lbrace A,B \rbrace$ given that $\vert \nu \vert = \Delta$.} \newline

\noindent
In this case, the adversary proceeds as follows:
\begin{enumerate}
    \item Assume the distribution centred on $0$ is the \say{negative} class ($\mathcal{H}_0$) and the one at distance $\Delta$ is the positive class ($\mathcal{H}_1$).
    \item Collect $y_1, \cdots, y_N$ observations under composition.
    \item Compute  $T(y) = \left \vert \frac{1}{N}\sum_{i=1}^Ny_i \right \vert^2$.
    \item Set a threshold $c$ to match a desired power and significance level (or $P_d$/$P_f$). If $T(y) < c^2$, assign to the negative class (fail to reject the null), else to the positive class (reject the null).
    \item Repeat the process, but consider $\Delta$ the negative class and $0$ the positive one. 
\end{enumerate}

For our experiment, we flip a fair coin and sample either $y \sim \mathcal{N}(0, \sigma^2)$ or $y \sim \mathcal{N}(\Delta, \sigma^2)$ to reflect an indifferent prior. The process is repeated $1000$ times and at $100$ values of $c^2$. At each grid point, we measure the true state of the system and the adversary's decision and use them to compute the empirical ROC curves, which we compare to the expected curves from Theorem 1. Figure \ref{fig:predicted_actual} shows the empirical ROC curves in comparison with their theoretical counterparts whereas Figure \ref{fig:actual_symmetric_gaussian} shows them in comparison to the symmetrified ROC and the NPO ROC curve. The empirical ROC curves match the expected ones very well, indicating that our bounds are tight and that the GLRT relaxation indeed affords higher privacy than the NPO assumption.

\begin{figure}[h]
    \centering
    \includegraphics[width=0.45\textwidth]{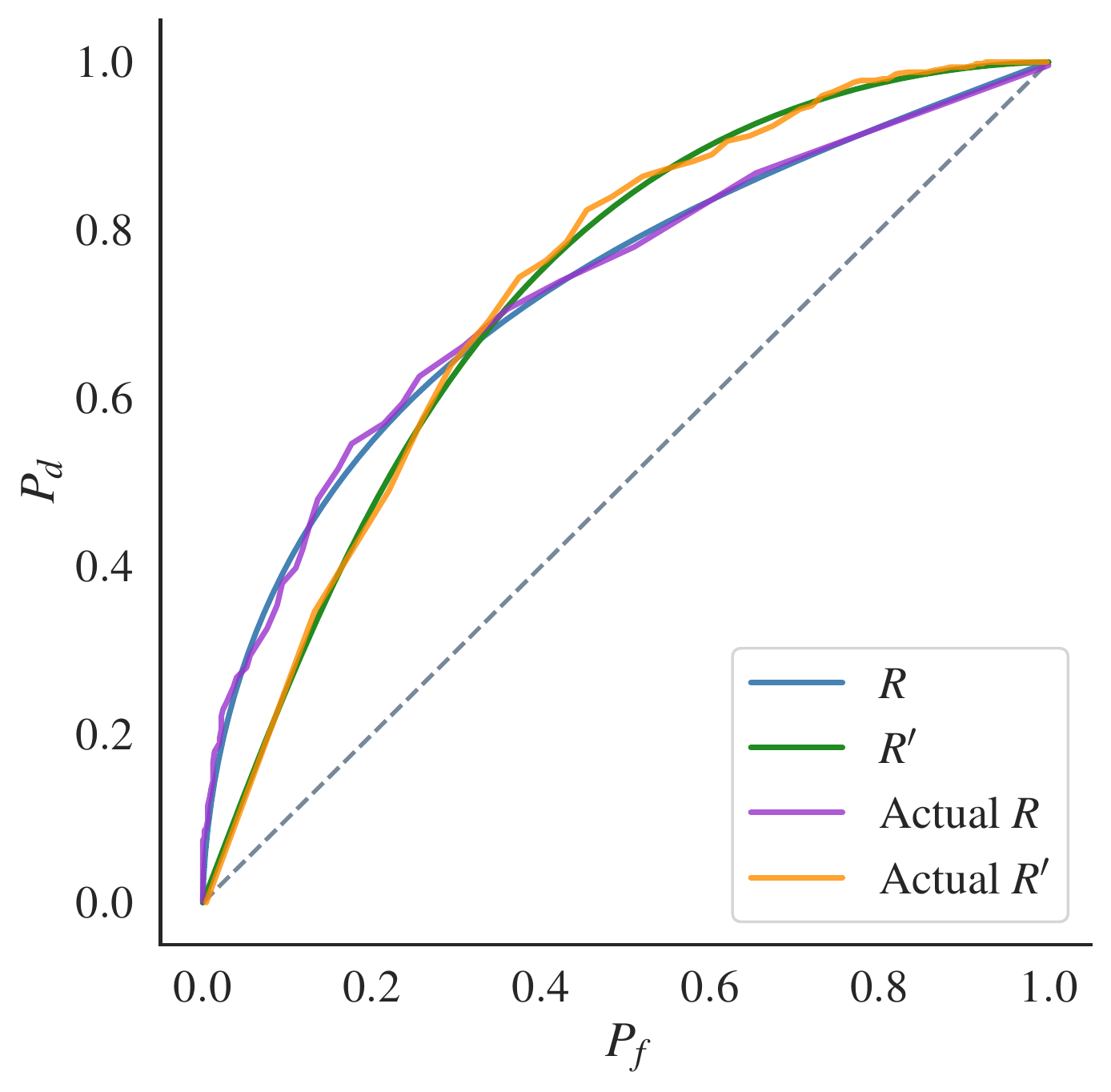}
    \caption{Empirical (purple and orange) and theoretical (blue and green) ROC curves for the worst-case adversary. $\Delta=1, \sigma^2=36, N=70$ composition rounds. The empirical curves match the theoretical ones. Average of $1000$ experiments and $100$ cutoff values. $(\varepsilon, \delta) = (2.94, 10^{-2})$.}
    \label{fig:predicted_actual}
\end{figure}

\begin{figure}[h]
    \centering
    \includegraphics[width=0.45\textwidth]{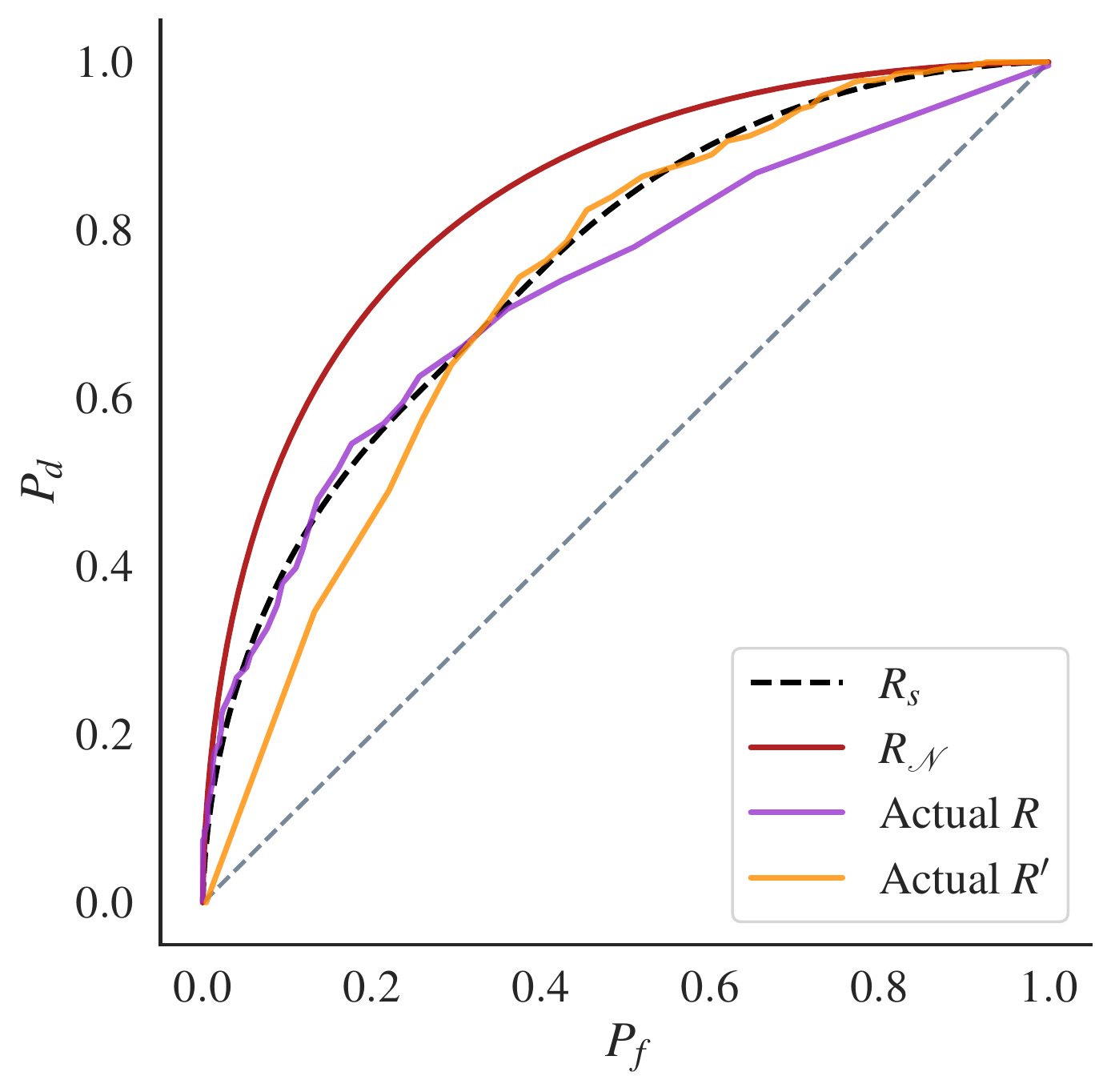}
    \caption{The same empirical (purple and orange) ROC curves for the experiment in Figure \ref{fig:predicted_actual}, but with the symmetrified theoretical curve (dashed black) and the NPO Gaussian mechanism curves (red) shown for comparison. The GLRT setting affords more empirical and theoretical privacy than the NPO regime. $(\varepsilon, \delta) = (2.94, 10^{-2})$ for the GLRT setting and $(3.63, 10^{-2})$ for the NPO setting. }
    \label{fig:actual_symmetric_gaussian}
\end{figure}

\subsection{Empirical performance in the high-dimensional query case}
To evaluate the theoretical claim that high-dimensional queries afford stronger privacy, we repeated the aforementioned experimental setting with the following modifications. Instead of a scalar query, the adversary now receives a $50$-dimensional vector with \textit{no capability to reduce the effective dimensionality} e.g. through auxiliary information. We assume $N=50$ composition rounds and $\Delta=1, \sigma^2=12.25$. As witnessed in Figure \ref{fig:high_dimensional}, the classification problem is indeed considerably harder for the adversary in this setting compared to the scalar problem above.

\begin{figure}[h]
    \centering
    \includegraphics[width=0.45\textwidth]{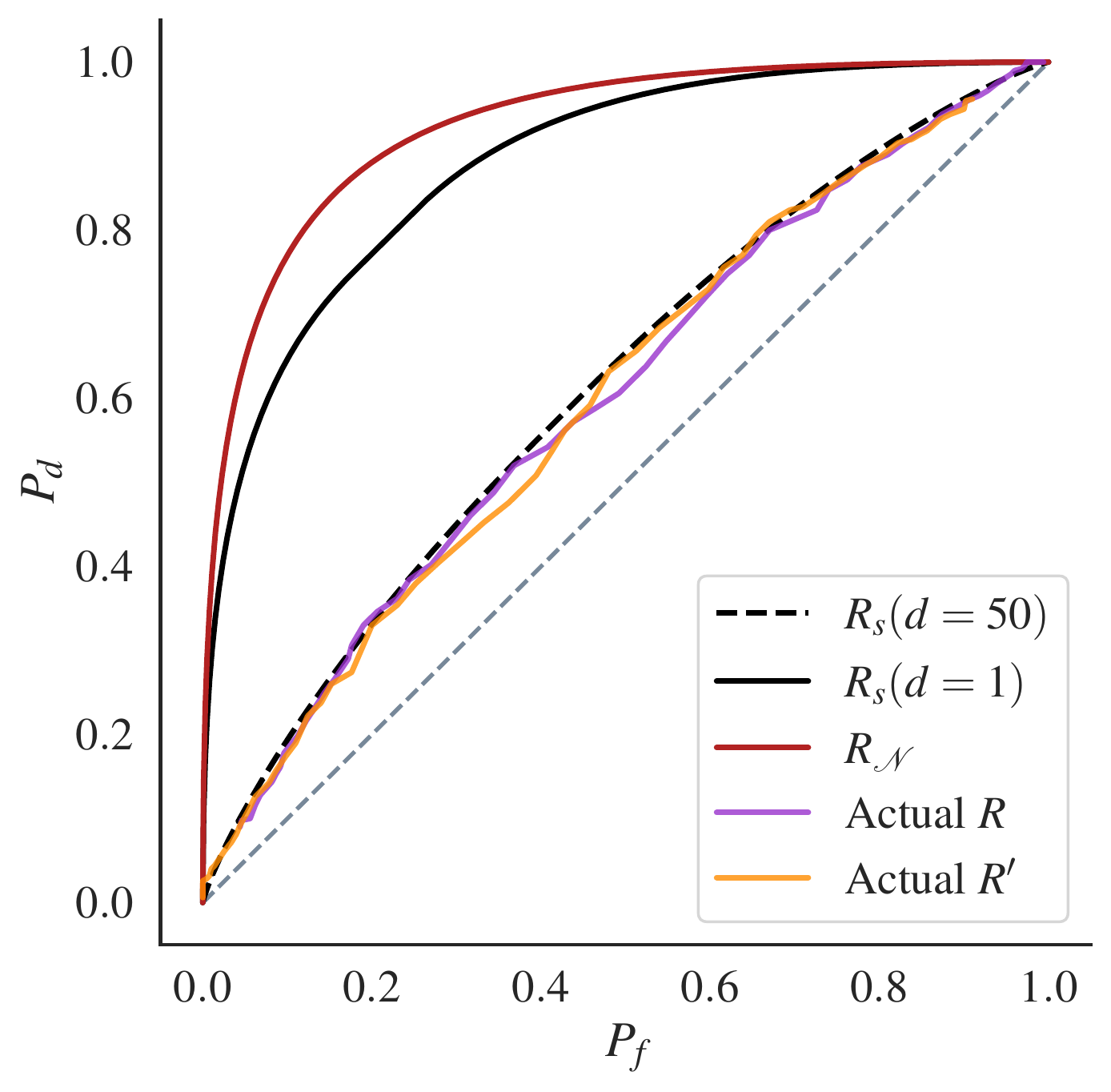}
    \caption{Empirical ROC curves (purple and orange) for the high-dimensional query experiment. Compared to the expected ROC curve for the same dimensionality (dashed black), the empirical result matches exactly. However, compared to the scalar case in the same setting (continuous black), the hypothesis test is considerably harder. The NPO adversary (red curve) has the easiest task. $(\varepsilon, \delta) = (0.76, 10^{-2})$ for $d=50$, $(5.39, 10^{-2})$ for $d=1$ and $(6.01, 10^{-2})$ for the NPO setting.}
    \label{fig:high_dimensional}
\end{figure}

\section{Conclusion}
The broad application of differential privacy will require our community to address legitimate stakeholder concerns about the privacy-utility trade-offs of privacy-preserving systems. Several prior works have empirically noted that, under realistic conditions, the guarantees of such systems are stronger than the worst case assumes. However, so far, there is a lack of investigations into \textit{formal} (rather than empirical) relaxations of the threat model and into the provision of tight guarantees in the spirit of DP. We find that the mild, and in many cases realistic, relaxation from an NPO to a GLRT adversary yields substantially amplified privacy bounds for the Gaussian mechanism. This is especially true for strongly private mechanisms with low sensitivity and high noise magnitudes. 

Our work should not be misconstrued as an attempt to undermine the DP guarantee. We explicitly request that our guarantees be given \textit{alongside} the worst-case guarantees and not in isolation, so as not to mislead stakeholders but instead to inform them more comprehensively about the privacy properties of sensitive data-processing systems. 

\printbibliography

\end{document}